\documentclass[runningheads,envcountsame]{llncs}

\usepackage[T1]{fontenc}
\usepackage{graphicx}

\usepackage{amsmath}
\usepackage{amssymb}
\usepackage{booktabs}

\usepackage{algorithm}
\usepackage{algpseudocode}

\usepackage{url}
\let\oldunderscore\_
\renewcommand{\_}{\oldunderscore\allowbreak}
\usepackage[hidelinks,bookmarks=false]{hyperref}

\newenvironment{restatedlemma}[1]{%
  \par\medskip
  \noindent\textbf{Lemma~\ref{#1}.}\ \itshape
}{%
  \par\medskip
}

\spnewtheorem{observation}[theorem]{Observation}{\bfseries}{\itshape}

\begin{document}
\title{A rooted tree framework for linear time ultrabubble detection}

\author{Athanasios E. Zisis\inst{1}\orcidID{0009-0000-8691-0993}\\
P{\aa}l S{\ae}trom\inst{1,2,3}\orcidID{0000-0001-8142-7441}}
\authorrunning{A. E. Zisis and P. S{\ae}trom}
\institute{
Department of Computer Science,
Norwegian University of Science and Technology,
Sem S{\ae}lands vei 9, 7034 Trondheim, Norway\\
\email{athanas.zisis@gmail.com, pal.satrom@ntnu.no }
\and
Department of Clinical and Molecular Medicine,
Norwegian University of Science and Technology,
Erling Skjalgsons gate 1, 7491 Trondheim, Norway\\
%\email{pal.satrom@ntnu.no}
\and
Sentral Stab,
St. Olavs Hospital HF,
7006 Trondheim, Norway
}

\maketitle              
\begin{abstract}
Pangenomics uses graphs to show genetic differences within or between species. In these graphs, a path can represent one genome, while regions with different paths show genetic variation. Biedged graphs use black edges for sequences and grey edges for links between them. Snarls are minimal subgraphs of a biedged graph that are separated from the rest of the graph by removing two black edges. Ultrabubbles are minimal acyclic and tip-free snarls and thus are important variant structures because they have finite paths and lack dead ends. 

In our previous work, we showed that in linear time every bidirected graph can be transformed to a rooted biedged bipartite one, and that in these graphs, ultrabubbles can be enumerated with a lowest common ancestor (LCA)-based method in $O(Kn)$ time, where $n$ and $K$ are the number of nodes and given snarls, respectively, of the graph. 

Here, we present a series of practical and theoretical improvements to our previous LCA-based approach. First, we present a hybrid method that selects between the LCA-based method and the naive approach for evaluating a snarl, depending on the size of the snarl in relation to the number of tips and cycle-closing nodes in the graph. Second, by using the theoretical framework from our previous paper, we show that all ultrabubbles can be found in $O(n+m+K)$ time, where $m$ is the number of edges, by traversing the breadth-first search (BFS) tree of the biedged bipartite graph. Third, we show that any two snarls that are candidate ultrabubbles and share a frontier node cannot be ultrabubbles; 
the resulting set of snarls is compatible, bound by $n$, and defines exclusive families of nested snarls.

We combine these three results into six methods and present benchmarking results that illustrate how the above improvements affect practical run-times for identifying ultrabubbles.

\end{abstract}

\section{Introduction}
\label{sec:introduction}

Variation graphs are important data structures in modern genomics because they encode both the genome and the genetic variation within a population. Genetic variation can be recognized as bubble-like structures in the graph: subgraphs that have more than one path between a pair of distinct nodes, the source and the sink. Among such bubbles, ultrabubbles are of particular interest as they are minimal acyclic subgraphs with a single source and a single sink, and thereby represent a finite number of genetic variants within a distinct part of the corresponding genome.

This paper is an improvement of our previous work based on lowest common ancestor (LCA) queries for ultrabubble enumeration \cite{zisis2026ultrabubble} in a biedged bipartite graph representation of variation graphs.

This previous work described a linear time transformation of a bidirected variation graph into a biedged bipartite graph $B$, where sequences are encoded along special black edges. Then, given the set of minimal 2-black-edge-connected subgraphs (so-called snarls) of $B$, 
we presented an LCA-based approach to enumerate the ultrabubbles of $B$ by answering LCA queries in a breadth-first search (BFS) tree $B_t$ of $B$.

As answering an LCA query takes constant time, the algorithm, which we here call UltraLCA, has 
an $O(Kn)$ complexity, improving the naive algorithm's $O(K(n+m))$ \cite{paten2018superbubbles}, where $K$, $n$, and $m$ are the number of snarls, nodes, and edges of $B$, respectively. 

Recently, three linear time approaches were introduced for ultrabubble enumeration. One approach used a specific traversal in certain bidirected graphs to run directed graph algorithms and then executed a weak superbubble enumeration algorithm to enumerate the ultrabubbles \cite{harviainen2026scalable}. The second approach is a tool that computes in linear time both snarls and ultrabubbles directly in bidirected graphs by using an SPQR tree framework \cite{sena2026identifying}. The third approach transforms a bidirected graph into a directed graph using the graph doubling technique, and then enumerates the ultrabubbles by computing the weak superbubbles in this doubled directed graph \cite{schmidt2026powergraphdoublingcomputing}.

Whereas the worst case run time of the UltraLCA algorithm is $O(Kn)$, its practical run time depends on the number of cycle-closing nodes and nodes with no outgoing edges (so-called tips) in the graph. This is because UltraLCA runs, for each snarl, two LCA queries for each tip and cycle-closing node. 

In this paper, we significantly improve the performance of UltraLCA. We also give a linear time of $O(n+m+K)$ approach for enumerating ultrabubbles, taking advantage of the same LCA framework and theory used in \cite{zisis2026ultrabubble}. Specifically, we first describe an approach that uses a sweep-scan from the root to the end of the BFS tree $B_t$ of the graph $B$. We then show how this sweep-scan approach can be combined with the LCA framework of \cite{zisis2026ultrabubble} to achieve a linear-time $O(n+m+K)$ complexity for enumerating ultrabubbles.

We also provide new hybrid algorithms that combine different algorithmic strengths, and nested family algorithms that take advantage of the nested structure of snarls.

As the linear time LCA-based approach for enumerating ultrabubbles presented here is different from the other linear time approaches \cite{harviainen2026scalable,sena2026identifying,schmidt2026powergraphdoublingcomputing}, we expect these methods to show different pros and cons in different practical situations.

The rest of this paper is structured as follows. Section~\ref{subsec:preliminaries} provides the preliminaries, Section~\ref{sec:algorithms} presents the algorithms, Section~\ref{sec:methods} describes the materials and methods, Section~\ref{sec:results} reports the results, Section~\ref{sec:discussion} discusses the results and concludes. Additional details on the preprocessing of the proposed algorithms and omitted proofs are given in the %Appendix~\ref{sec:appendix}.
Appendices A and B.

\subsection{Preliminaries}
\label{subsec:preliminaries}
Since the present paper is an improvement of \cite{zisis2026ultrabubble} we recall only the definitions and results needed here.

A bidirected graph where nodes represent homologous DNA sequence regions and edges their adjacencies is called a \textit{variation graph} \cite{paten2018superbubbles}.%, \cite{zisis2026ultrabubble}.

A \textit{biedged graph} is constructed from a variation graph
by splitting each of the nodes $v$ in two, $v_{left}, v_{right}$, and connecting them with a black edge. The variation graph's original edges, which were connected to either the left or right side of $v$, are kept and colored gray \cite{paten2018superbubbles}.

A \textit{snarl} \cite{paten2018superbubbles} $x,y$ is a minimal subgraph of a biedged graph that satisfies the following conditions:
\begin{itemize}
    \item \text{Separability:} After removing the distinct black edges incident to $x$ and $y$, the graph is split into disconnected components. Both $x$ and $y$ belong to a common component $C$ while the opposite endpoints of the corresponding black edges, denoted by $x'$ and $y'$ respectively, do not belong to $C$. 
    \item \text{Minimality:} There is no other black edge $(v,v')$ in the snarl's subgraph such that both $(x,v)$ and $(v',y)$ satisfy the above separability criterion.
\end{itemize}

We define \textit{crossing snarls} as two or more snarls that share some nodes, are not nested to each other, and do not have any common frontier nodes.

A \textit{tip} in a biedged graph is a leaf node that has no incident gray edges. A tip represents a dead end of a biedged graph.

An \textit{ultrabubble} is a snarl without cycles or tips. 
Here, we assume that every given bidirected graph has been transformed in linear time as detailed in \cite{rautiainen2020graphaligner} and \cite{zisis2026ultrabubble} to a biedged bipartite one, which follows specific traversal rules as described below. Note that the above transformation maintains the connectivity, all path information, and cycles of the initial given bidirected graph and thus does not affect the snarl set of the graph \cite{rautiainen2020graphaligner}, \cite{zisis2026ultrabubble}. However, adding an artificial root, if needed, could alter the snarl set, while the ultrabubble set remains unaffected, as illustrated in \cite{zisis2026ultrabubble}. 

A biedged bipartite graph is a directed biedged graph that has two types of nodes: L (left) and R (right). All black edges go from an L node to an R node, whereas all grey edges go from an R node to an L node.

A traversal in a biedged graph alternates between black and grey edges \cite{paten2018superbubbles}. In this paper, with the transformation to a biedged bipartite graph we apply, and following the traversal rules described above, we define a path as a sequence of connected and non-repeating nodes that follows these traversal rules and a cycle as a closed path that starts and ends at the same node. Since these biedged bipartite graphs are isomorphic to directed graphs \cite{zisis2026ultrabubble}, these path and cycle definitions here correspond to those in directed graphs while following the biedged bipartite structure and the traversal rules described earlier.

We assume that every biedged graph we are dealing with has a working root, that is, a node that we can start from and be able to traverse all nodes of the graph following the traversal rules (edge directions) we explained above. If such a root does not exist, one can be constructed in linear time \cite{zisis2026ultrabubble}.
Given the above setup, it is clear that we have adapted a left-right orientation of the graph from the root to the end (sink). Among the  two nodes of a black edge,  the L node is the one closest to the root.

All graphs that we are dealing with in this paper have the above characteristics and from now on, we use the term biedged graph $B$ to mean a graph that has all the above characteristics; that is, $B$ has a root, a left-right orientation, an end (sink) node, and the specific edge directions mentioned above.

Given such a biedged graph $B$, the following three observations apply to snarls and ultrabubbles. First, only $R-L$ snarls –– that is, snarls $(x,y)$ where $x$ is an R node, $y$ is an L node, and $x$ is closest to the root –– can be ultrabubbles \cite{zisis2026ultrabubble}. Second, a depth first search (DFS) from the root in $B$ enumerates all tips and cycle-closing nodes in $B$; a cycle-closing node is part of a snarl iff the corresponding cycle is completely nested within the snarl's subgraph \cite{zisis2026ultrabubble}. Assuming that the set $ftip$ contains all tips and cycle-closing nodes in $B$, a snarl $(x,y)$ is therefore an ultrabubble iff it is an $R-L$ snarl and none of the nodes in $ftip$ are present within its subgraph. Third, a node $t$ is part of an $R-L$ snarl $(x, y)$ iff $LCA_{B_t}(x, t) = x$ and $LCA_{B_t}(y, t) \neq y$, where $LCA_{B_t}(x, t)$ is the LCA of $x$ and $t$ in the BFS tree $B_t$ of $B$ \cite{zisis2026ultrabubble}. 

These observations lead to an algorithm that takes as input $ftip$ and $K$, the set of snarls of $B$, and for each $R-L$ snarl $k \in K$, runs up to $2*|ftip|$ LCA queries to determine if $k$ is an ultrabubble \cite{zisis2026ultrabubble}. This algorithm, which we name UltraLCA here, has complexity $O(|K|\cdot|ftip|)$, which is $O(|K|\cdot n)$ in worst case, where $n$ is the number of nodes of $B$.

In the following, we use $sn_1$ and $sn_2$ to denote the frontier nodes $(x, y)$ of an $R-L$ snarl.

\section{Algorithms}
\label{sec:algorithms}

This section first briefly explains two existing algorithms for identifying ultrabubbles in a given set of snarls within a biedged graph, the naive algorithm \cite{paten2018superbubbles} and the LCA based approach introduced in \cite{zisis2026ultrabubble}, which we here call UltraLCA (ULCA).

Then we introduce six new algorithms that aim to improve the LCA-based method. We group these algorithms into Sweep, Hybrid, and Nested snarl families algorithms and present each algorithm group, including underlying theory and proofs, in a separate section. 

\subsection{Existing algorithms}
\label{subsec:existing}
\subsubsection{Naive}
\label{subsec:naive}
Given the set of snarls of a biedged graph $B$, the Naive algorithm uses a DFS to check each snarl's subgraph for the presence of tips or cycles. 
Consequently, the algorithm has $O(K(n+m))$ complexity, where $K$, $n$, and $m$ are the number of snarls, nodes, and edges in $B$, respectively \cite{paten2018superbubbles}. 

\subsubsection{UltraLCA}
\label{subsec:algorithm1}
The UltraLCA algorithm uses two constant time LCA queries to determine if a node in the $ftip$ set is located within an $R-L$ snarl, resulting in an $O(Kn)$ complexity (see Section~\ref{subsec:preliminaries}) \cite{zisis2026ultrabubble}. Compared with the Naive algorithm, UltraLCA was faster on real-world data unless the $ftip$ set was large \cite{zisis2026ultrabubble}.

\subsection{Sweep algorithms} 
\label{subsec:sweepfamily}

The main disadvantage of the UltraLCA algorithm is that each node in the $ftip$ set is processed for each $R-L$ snarl in $B$. However, the following observation suggests an approach where the $ftip$ set can be gradually reduced.

\begin{observation}\label{o:depth}
  For any given $R-L$ snarl $(x, y)$ in $B$, as $y$ can only be reached through $x$, $d(x) < d(y)$, where $d(x)$ is the shortest path from the root of $B$ to $x$  \cite{zisis2026ultrabubble}. Consequently, for any node $v \in B$, $v$ is outside $(x, y)$ if $d(v) < d(x)$. 
\end{observation}

More specifically, given a DFS traversal of the BFS tree $B_t$ of $B$ and a node $v \in B$, let $t_d(v)$ and $t_f(v)$ be the discovery and finishing times, respectively, for $v$ in the DFS traversal. The following two observations then hold for any given $R-L$ snarl $(x, y)$ in $B$ and any node $v \in B$:

\begin{observation}\label{o:discovery}
   If $t_d(v) < t_d(x)$ then $v$ is outside $(x, y)$.
\end{observation}

\begin{observation}\label{o:finish}
   If $t_d(v) > t_f(x)$ then $v$ is outside $(x, y)$.
\end{observation}

Observation~\ref{o:discovery} follows from Observation~\ref{o:depth}; Observation~\ref{o:finish} further follows from the properties of DFS and $R-L$ snarls, as $t_f(x) > t_f(y)$ for any $R-L$ snarl \cite{zisis2026ultrabubble}. Together, these observations suggest that the candidate snarls and the nodes in $ftip$ should be processed in order of $t_d$, progressively removing from $ftip$ the nodes that are guaranteed to be outside the remaining snarls. 

\begin{algorithm}[t]
\caption{SWEEP}\label{alg:sweep}
\begin{algorithmic}[1]
\Require The set $S_{RL}$ of $R-L$ snarls in $B$, the set $ftip$, and the BFS tree $B_t$ of $B$.
\Require $t_d(v)$: the discovery time of each node $v \in B$ in a DFS search of $B_t$.
\Require $t_f(v)$: the finishing time of each node $v \in B$ in a DFS search of $B_t$.
\Require $LCA_{B_t}(a,b)$: Lowest common ancestor of $a$ and $b$ in $B_t$.

\State $U \gets \emptyset$
\State $T \gets ftip$
\For{each $(x,y) \in S_{RL}$ in order of $t_d(x)$}
    \State $\mathit{F} \gets \textbf{true}$
    \For{each $t \in T$}
        \If{$t_d(t) < t_d(x)$}
            \State $T \gets T\setminus t$
        \ElsIf{$t_d(t) < t_f(x)$ and $LCA_{B_t}(t,x) = x$ and $LCA_{B_t}(t,y) \ne y$}
            \State $\mathit{F} \gets \textbf{false}$
            \State \textbf{break}
        \EndIf
    \EndFor
    \If{$\mathit{F} = \textbf{true}$}
        \State $U \gets U \cup \{(x,y)\}$
    \EndIf
\EndFor
\State \textbf{return} $U$
\end{algorithmic}
\end{algorithm}

\begin{figure}[ht!]
\centering
\includegraphics[width=\textwidth]{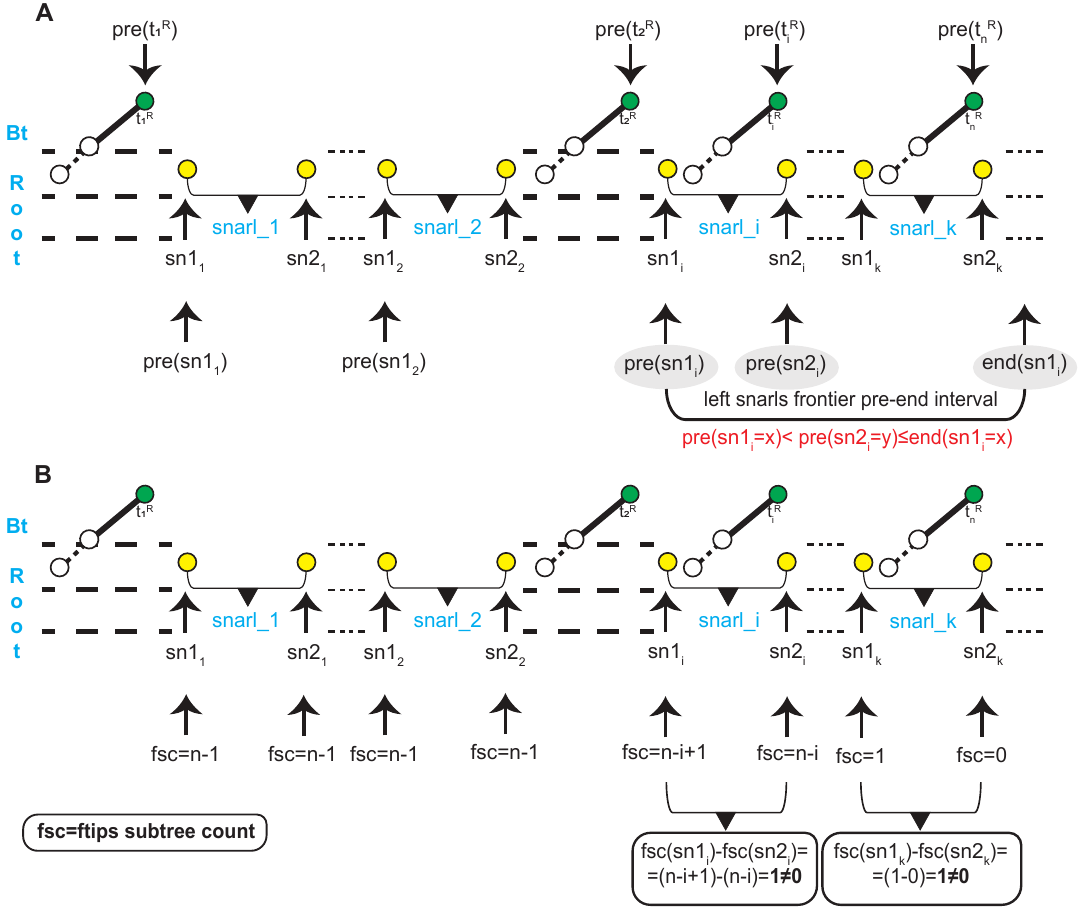}
\caption{A DFS traversal at the BFS tree $B_t$ of the biedged graph $B$, computes the preorder index ("pre", or discovery time $t_d$) and finishing position of the subtree ("end", or finishing time $t_f$) of every node of the graph. Dashed lines indicate omitted parts of $B_t$ between the illustrated structures. \textbf{(A)} %For every node $x$ of $B_t$, its preorder position pre($x$) ($t_d(x)$) and  finishing position end($x$) ($t_f(x)$), is computed. 
By the $R-L$ snarl definition, for each $R-L$ snarl $(sn_1,sn_2)$, \(pre(sn_1)<pre(sn_2) \leq end(sn_1)\).
\textbf{(B)} The $ftip$ subtree count (fsc) of a node $x$ is the number of nodes in $ftip$ that lie in the interval of (pre($x$),end($x$)). For a snarl $(sn_1,sn_2)$, if the subtraction \( fsc(sn_1)-fsc(sn_2) \neq 0\), then at least one node from the $ftip$ set lies between the snarl's frontier nodes. Consequently, that snarl cannot be an ultrabubble.}
\label{fig:sweep}
\end{figure}

The SWEEP (SW) algorithm (Algorithm~\ref{alg:sweep}) implements this ordered processing and uses the LCA queries of the UltraLCA algorithm only for the nodes in $ftip$ that reside in the time interval defined by $[t_d(x), t_f(x)]$. SWEEP complexity is still $O(Kn)$, as each snarl may scan up to $|ftip|$ active points, where $|ftip| \leq n$, and $K$ and $n$ are the number of snarls and nodes of $B$, respectively. 

Given $c(t)$, the number of nodes in $ftip$ discovered during or after time point $t$ in a DFS of $B_t$, this ordered processing can be further improved through the following two observations:
\begin{observation}\label{o:cumsum}
  Given a node $v$ in $B_t$, the number of nodes from $ftip$ that are present in the subtree of $v$ in $B_t$ is $fsc(v) = c(t_d(v))-c(t_f(v))$.
\end{observation}

\begin{observation}\label{o:interval}
  For any given $R-L$ snarl $(x, y)$ in $B$, the number of nodes from $ftip$ that are present in $(x, y)$ is $fsc(x) - fsc(y)$ (Figure~\ref{fig:sweep}).
\end{observation}

\begin{algorithm}[t]
\caption{SWEEP\_INTERVAL}\label{alg:sweep_interval}
\begin{algorithmic}[1]
\Require The set $S_{RL}$ of $R-L$ snarls in $B$, the set $ftip$, and the BFS tree $B_t$ of $B$.
\Require $t_d(v)$: the discovery time of $v$ in a DFS search of $B_t$.
\Require $t_f(v)$: the finishing time of $v$ in a DFS search of $B_t$.
\Require $c(t)$: the number of nodes in $ftip$ discovered during or after time point $t$ in a DFS of $B_t$.

\State $U \gets \emptyset$
\For{each $(x,y) \in S_{RL}$}
    \If{$(c(t_d(x))-c(t_f(x)))-(c(t_d(y))-c(t_f(y))) = 0$}
        \State $U \gets U \cup \{(x,y)\}$
    \EndIf
\EndFor
\State \textbf{return} $U$
\end{algorithmic}
\end{algorithm}

In the resulting SWEEP\_INTERVAL (SWI) algorithm (Algorithm~\ref{alg:sweep_interval}), the running time per snarl is constant. As the algorithm also depends on the BFS and DFS, which are $O(n+m)$, the total complexity of SWEEP\_INTERVAL is $O(n+m+K)$, where $n$, $m$, and $K$ are the number of nodes, edges, and snarls of $B$ .

\subsection{Hybrid algorithms}
\label{subsec:hybridfamily}
Hybrid algorithms use the idea of sending each $R-L$ snarl $(x,y)$ to the most appropriate, feasible, and meaningful algorithm for processing. In the next two subsections, we present hybrid algorithms that choose online which algorithm to process each snarl. The criterion for choosing which algorithm to use is the number of nodes and edges of each snarl in relation to the number of LCA queries needed by the UltraLCA approach to evaluate the snarl.

\subsubsection{Hybrid UltraLCA-Naive online}
\label{subsec:hybridalgo1}
HYBRID\_UltraLCA\_ONLINE (H-ULCA) is an online hybrid algorithm that decides for each $R-L$ snarl $(x,y)$ whether it will be processed by the Naive algorithm or UltraLCA. Specifically, any snarl $S$ should be processed by the LCA queries of UltraLCA, if $n_S + m_S > \alpha\cdot|ftip|$, where $n_S$ and $m_S$ are the number of nodes and edges in the snarl, respectively, and $\alpha$ is a constant. Importantly, as the size of each snarl is unknown, HYBRID\_UltraLCA\_ONLINE starts processing each snarl with the Naive algorithm, but switches to the UltraLCA approach if the $\alpha\cdot|ftip|$ threshold is reached. Empirical evaluations suggested $\alpha=1.7$ (data not shown). The complexity is $O(K(n+m))$, the same as the Naive algorithm.

\subsubsection{Hybrid Sweep-Naive online}
\label{subsec:hybridsweep}
HYBRID\_SWEEP\_ONLINE (H-SW) follows the same logic as in the previous Section, but diverts the snarl to the SWEEP algorithm instead of UltraLCA if $n_S + m_S > \alpha\cdot|ftip|$. Also here, the complexity is $O(K(n+m))$.

\subsection{Nested snarl families algorithms}
\label{subsec:nestedfamily}
Snarls can be nested such that their subgraphs share some of the same nodes and edges.
Here, the idea is to take advantage of the nested structure hierarchy of so-called compatible snarls to avoid unnecessary recalculations, as we explain in the following.

A set of snarls $G$ is called compatible if no two snarls of the set partially overlap; that is, their subgraphs are separate or one is contained within the other \cite{paten2018superbubbles}.

\begin{lemma}
\label{lem:overlap}
In the set $S_{RL}$ of $R-L$ snarls $(x,y)$ of a given rooted biedged graph $B$, no pair of crossing snarls can exist.
\end{lemma}

\begin{lemma}
\label{lem:leftmost} 
Given the set $S_{RL}$ of $R-L$ snarls $(x,y)$ of a rooted biedged graph $B$, any snarls in $S_{RL}$ that share the leftmost frontier node cannot be ultrabubbles.
\end{lemma}

\begin{lemma}
\label{lem:rightmost}
Given the set $S_{RL}$ of $R-L$ snarls $(x,y)$ of a rooted biedged graph $B$, no two distinct snarls in $S_{RL}$ will share their rightmost frontier node $y$. 
\end{lemma}

By Lemmas~\ref{lem:overlap}, \ref{lem:leftmost}, \ref{lem:rightmost} (see Appendix~\ref{sec:omitted_proofs} for proofs) we have Theorem~\ref{thm:compatible}.

\begin{theorem}
\label{thm:compatible}
 Given the set $S_{RL}$ of $R-L$ snarls $(x,y)$ of a given rooted biedged graph $B$, after removing the $R-L$ snarls that share a common frontier node, we get the set of compatible $R-L$ snarls of $B$, $K_c$. All ultrabubbles in $B$ will be members of $K_c$. 
\end{theorem}
  Clearly, $K_c$ is bounded by $n$, where $n$ is the number of nodes of $B$.

\begin{observation}
 The aforementioned Lemmas~\ref{lem:overlap}, \ref{lem:leftmost}, \ref{lem:rightmost} and Theorem~\ref{thm:compatible} are in alignment and correspondence with the established theory for compatible families of snarls \cite{paten2018superbubbles}, but here proofs are simpler because of the different topology of the rooted biedged bipartite graphs we work with. Specifically, there is no need of the bridge edges approach that is used in \cite{paten2018superbubbles}.    
\end{observation}
  
Given the $R-L$ snarls $(x,y)$ of a rooted biedged graph $B$, by Theorem~\ref{thm:compatible}, we can obtain $K_c$ by removing the snarls that share a frontier node. As $K_c$ is compatible, for each pair of snarls of $K_c$, the pair is either nested to each other or they have separate subgraphs. We can now use the following Lemma to identify the nested structure of $K_c$.

\begin{lemma}
\label{lem:nested_finish}
Given a DFS traversal of the BFS tree of a rooted biedged graph $B$, the corresponding finishing times $t_f(v)$ for each node $v \in B$, the set $K_{c}$ of compatible $R-L$ snarls of $B$, and two snarls $\{(r_i, l_i), (r_j, l_j)\} \in K_c, i\ne j$, $(r_j, l_j)$ is nested within $(r_i, l_i)$ iff $t_f(r_i) > t_f(r_j) > t_f(l_j) > t_f(l_i)$. 
\end{lemma}
Given $K_c$ and Lemma~\ref{lem:nested_finish} (see Appendix~\ref{sec:omitted_proofs} for proof), we can now find the nested hierarchy of parent-child snarls by adding snarl $i$ as a new node and moving below that node in the hierarchy when finishing node $l_i$ in the DFS, and moving up one level in the hierarchy when finishing an $r$ node (for any $r$ node in $K_c$). This procedure creates a set of trees, where each tree represents a nested hierarchy of parent-child snarls.  

Given this nested hierarchy, one can traverse the snarls either bottom-up or top-down to identify the snarls that are ultrabubbles. Traversing the hierarchy bottom-up corresponds to processing the inner snarls before the outer ones. In this case, if snarl $i$ contains a node from the $ftip$ set, we can immediately reject both snarl $i$ and its parents as ultrabubbles. However, all of snarl $i$'s siblings must still be evaluated. Traversing the hierarchy top-down corresponds to processing the outer snarls before the inner ones. In this case, if snarl $i$ is identified as an ultrabubble, then all of $i$'s children are also ultrabubbles. Of the two, the top-down approach is the more appropriate if most candidate snarls are ultrabubbles; vice versa if most snarls are rejected as ultrabubbles. Here, we choose to use the top-down method, as implemented in the following two algorithms, NESTED\_UltraLCA (N-ULCA) and NESTED\_HYBRID\_UltraLCA (NH-ULCA).
NESTED\_UltraLCA  uses UltraLCA to check the nested snarls top-down. If a parent is accepted as an ultrabubble, then also all of its child snarls are accepted; if not, the algorithm recursively checks all its children according to the hierarchical structure.
NESTED\_HYBRID\_UltraLCA follows the same logic as NESTED\_UltraLCA, but chooses between UltraLCA and the Naive algorithm as described in Section~\ref{subsec:hybridalgo1}. 

\section{Materials and methods}
\label{sec:methods}

\subsection{Graphs}
\label{subsec:graphs}
We used bidirected variation graphs in Graphical Fragment Assembly (GFA) format, both real and synthetic ones, that were processed as described in this paper and in \cite{zisis2026ultrabubble}, to become biedged bipartite graphs isomorphic with directed ones. The graphs used here are the same as in \cite{zisis2026ultrabubble} and some basic information about them is provided in Table~\ref{tab:graphinfo}.

\begin{table}[ht!]
\centering
\caption{Graphs used for evaluation purposes in our experiments. The graphs were preprocessed as described in \cite{zisis2026ultrabubble}.}
\label{tab:graphinfo}

\resizebox{\textwidth}{!}{%
\begin{tabular}{lccc}
\toprule
Dataset & NODES & EDGES & TIPS \\
\midrule

\path{C4-90_VG_rem_nroot} & 19 & 20 & 7 \\

\path{cerevisiae.fa.gz.d1a145e.417fcdf.7493449.smooth.final_comp11_rem_nroot}
& 54721 & 73869 & 9 \\

\path{cerevisiae.fa.gz.d1a145e.417fcdf.7493449.smooth.final_comp4_nroot}
& 72349 & 97855 & 7 \\

\path{cerevisiae.fa.gz.d1a145e.417fcdf.7493449.smooth.final_comp7_rem_nroot}
& 216342 & 292679 & 55 \\

\path{chr19.pan.fa.a2fb268.e820cd3.9ea71d8.smooth_comp14_rem_nroot}
& 2858719 & 3985788 & 458 \\

\path{chr6.C4} & 1748 & 2366 & 2 \\

\path{chr6.pan.fa.a2fb268.2ff309f.1300c8a.smooth_comp1_rem_nroot}
& 4906290 & 6856764 & 606 \\

\path{chr6.pan.fa.a2fb268.2ff309f.1300c8a.smooth_comp5_nroot}
& 5475 & 7374 & 2 \\

\path{chr6.pan.fa.a2fb268.2ff309f.1300c8a.smooth_comp6_rem_nroot}
& 1700 & 2283 & 4 \\

\path{chrM.pan.fa.6626ff2.7748b33.72587dd.smooth}
& 1395 & 1887 & 2 \\

\path{LPA} & 3751 & 5195 & 2 \\
\path{lpa120} & 120 & 169 & 2 \\
\path{MHC-57b_VG} & 1068 & 1496 & 2 \\

\midrule

\path{synth1} & 8091 & 10133 & 200 \\
\path{synth2} & 173879 & 219006 & 152 \\
\path{synth3} & 27886 & 35086 & 7 \\
\path{synth4} & 37105 & 47488 & 7 \\
\path{synth5} & 834549 & 1071462 & 157 \\
\path{synth6} & 211938 & 273297 & 3 \\

\bottomrule
\end{tabular}%
}
\end{table}

\subsection{Snarls set}
\label{subsec:snarls}
 We used the vg toolkit \cite{garrison2018variation} with the flag -T to identify snarls in the graphs. This tool outputs a hierarchical structure of compatible snarls, including the trivial ones, as explained in \cite{paten2018superbubbles}. 
 Note that the vg output is a subset of all potential snarls \cite{paten2018superbubbles}, 
 but the snarls missed by vg cannot form ultrabubbles \cite{paten2018superbubbles}. %, as shown in \cite{paten2018superbubbles} and \cite{zisis2026ultrabubble}. 
 Moreover, the JSON output files of vg identify which snarls are acyclic; we kept only the acyclic ones.
 
\subsection{Identifying ultrabubbles}
\label{subsec:idultra}
We used each of the eight algorithms described in Section~\ref{sec:algorithms} to identify the ultrabubbles in a given set of snarls. Several processing steps were shared between the algorithms and were implemented as separate functions; see Table~\ref{tab:algos} for an overview and Section~\ref{subsec:preprocessing} for details.

\subsection{Implementation setup}
\label{subsec:implementation}
 All algorithms were implemented in Python and the Python scripts were run on a laptop with Intel\textsuperscript{\textregistered} Core\textsuperscript{\texttrademark} i7-1280p that contained 14 Cores and 20 Threads and had 16 GB RAM. 
Codex (ChatGPT, OpenAI) was used to support code development.
The pipeline code is available at \url{https://github.com/athanasios-zisis/UltraLCA-2}.

\section{Results}
\label{sec:results}
We used a set of real variation graphs and synthetic biedged graphs \cite{zisis2026ultrabubble} to benchmark the six new algorithms 
SWEEP, SWEEP\_INTERVAL, HYBRID\_UltraLCA\_ONLINE, HYBRID\_SWEEP\_ONLINE, NESTED\_UltraLCA, and NESTED\_HYBRID\_UltraLCA against the UltraLCA \cite{zisis2026ultrabubble} and Naive \cite{paten2018superbubbles} algorithms; see Methods for details.

When evaluating the clean execution times without preprocessing, the sweep algorithms, SWEEP\_INTERVAL, SWEEP, and HYBRID\_SWEEP\_ONLINE, generally outperformed the other algorithms, with SWEEP\_INTERVAL being 10 to 100 times faster than the Naive algorithm (Table~\ref{tab:algoclean}, Figure~\ref{fig:cleanratio}). On the largest graphs that contained the most tips (CHR19-14 and CHR6-1), SWEEP\_INTERVAL was almost 1000 times faster than UltraLCA. The NESTED\_UltraLCA and UltraLCA algorithms had very similar execution times, with NESTED\_UltraLCA generally being slightly faster. The hybrid UltraLCA algorithms, HYBRID\_UltraLCA\_ONLINE and NESTED\_HYBRID\_UltraLCA, had execution times comparable to the Naive algorithm.  

When considering the execution times needed for preprocessing the graphs, the differences between the algorithms were smaller (Table~\ref{tab:pretotal}), as the algorithms relied on many of the same preprocessing methods. The main exceptions were the two nested snarl algorithms, which used about 50\% longer time for preprocessing than the UltraLCA algorithm. We note that apart from reading the graph and snarls, the Naive algorithm does no preprocessing, making it the fastest algorithm overall in this benchmark. However, as noted in \cite{zisis2026ultrabubble}, the Naive algorithm is much slower if it has to consider all potential snarls instead of only the $R-L$ snarls used as input here, as $L-L$ and $R-R$ snarls tend to have much larger subgraphs.

When considering the asymptotic complexities for the algorithms, the Naive algorithm is $O(K(n+m))$ but in sparse graphs where \(n \approx m\), such as those used in our benchmark, its complexity is $O(Kn)$. Consequently, of all the algorithms considered here, only the SWEEP\_INTERVAL algorithm's $O(n+m+K)$ complexity is better on sparse graphs. The SWEEP algorithm also has $O(Kn)$ complexity, but achieved practical execution times close to SWEEP\_INTERVAL on our benchmark, as the graphs typically were "long" instead of "broad" with many parallel paths, allowing nodes from the $ftip$ set to be eliminated from being evaluated multiple times.   

\begin{figure}[ht!]
\centering
\includegraphics[width=0.95\textwidth]{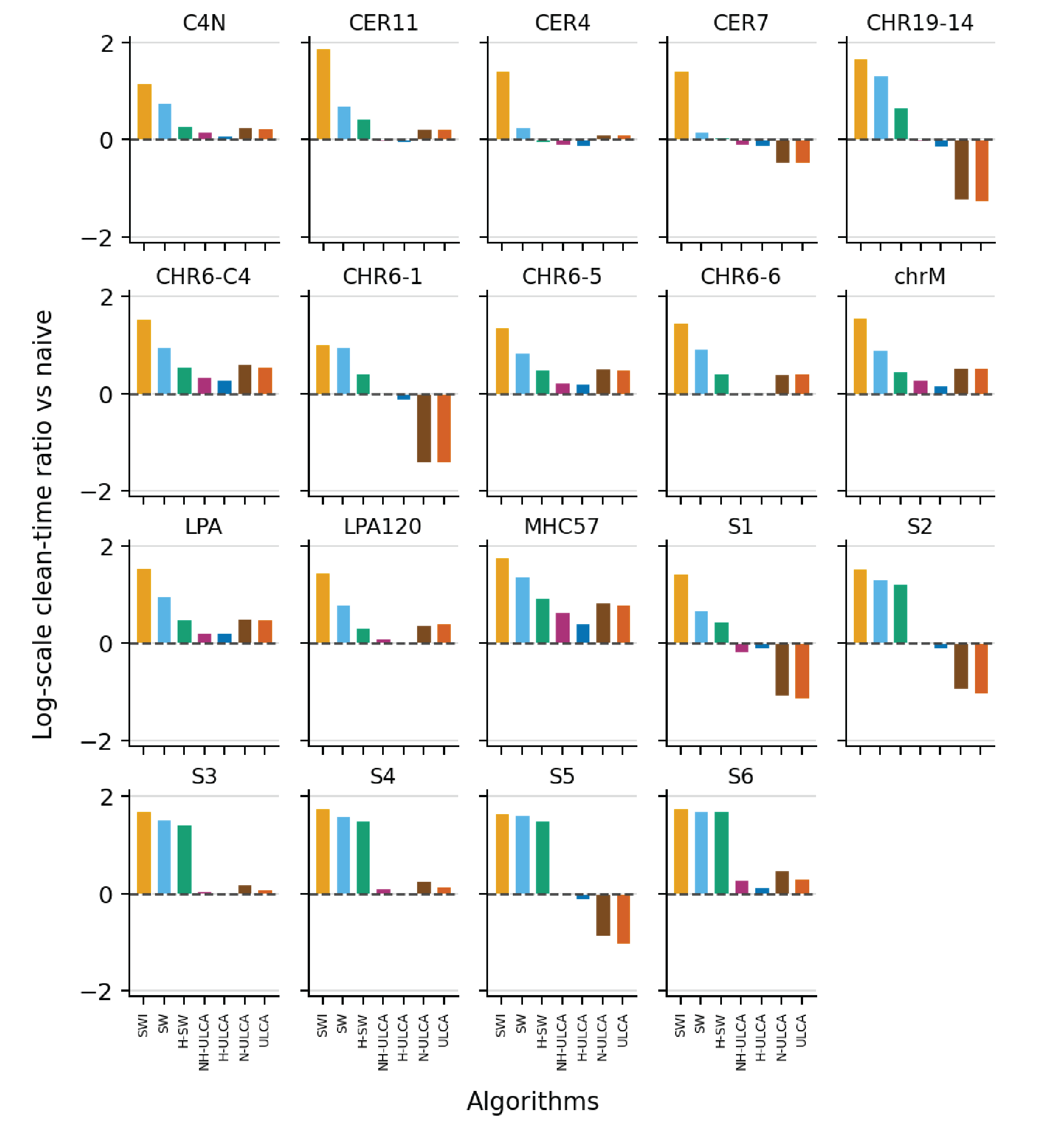}
\caption{Clean execution times per algorithm (x-axis) relative to those of the Naive algorithm on each dataset. The ratios (algorithm / Naive) are shown in $-\log_{10}$ scale (y-axis); consequently, a value above 0 means that the corresponding algorithm is faster than the Naive algorithm.
}
\label{fig:cleanratio}
\end{figure}

\begin{table}[ht!]
\centering
\caption{Clean execution times for running the algorithms post preprocessing on each dataset; see Table~\ref{tab:pretotal} for the preprocessing execution times. The columns show the dataset, the execution times for each algorithm (in seconds), and the number of input snarls, $R-L$ snarls, and ultrabubbles. The fastest execution time for each dataset is highlighted in bold.}
\label{tab:algoclean}
\resizebox{\textwidth}{!}{%
\begin{tabular}{lccccccccccc}
\toprule
Dataset & NAIVE & UltraLCA & SWEEP & \shortstack{SWEEP\\INTERVAL} & \shortstack{HYBRID\\UltraLCA\\ONLINE} & \shortstack{HYBRID\\SWEEP\\ONLINE} & \shortstack{NESTED\\UltraLCA} & \shortstack{NESTED\\HYBRID\\UltraLCA} & \shortstack{Total\\snarls} & \shortstack{R--L\\snarls} & Ultrabubbles \\
\midrule
C4-90\_VG\_rem\_nroot & 0.000122 & 0.000069 & 0.000021 & \textbf{0.000008} & 0.000102 & 0.000065 & 0.000066 & 0.000086 & 9 & 8 & 7 \\
cerevisiae.fa.gz.d1a145e.417fcdf.7493449.smooth.final\_comp11\_rem\_nroot & 0.322845 & 0.19535 & 0.064718 & \textbf{0.00417} & 0.364387 & 0.119781 & 0.193096 & 0.351836 & 18172 & 18169 & 18162 \\
cerevisiae.fa.gz.d1a145e.417fcdf.7493449.smooth.final\_comp4\_nroot & 0.358874 & 0.280119 & 0.19868 & \textbf{0.013702} & 0.492796 & 0.408072 & 0.280622 & 0.471028 & 23907 & 23905 & 23898 \\
cerevisiae.fa.gz.d1a145e.417fcdf.7493449.smooth.final\_comp7\_rem\_nroot & 1.154717 & 3.456369 & 0.780822 & \textbf{0.043667} & 1.555216 & 1.055648 & 3.529705 & 1.46166 & 71457 & 71443 & 71395 \\
chr19.pan.fa.a2fb268.e820cd3.9ea71d8.smooth\_comp14\_rem\_nroot & 17.936609 & 336.716343 & 0.87527 & \textbf{0.392153} & 24.81967 & 3.971431 & 299.633108 & 19.084889 & 904944 & 904930 & 904804 \\
chr6.C4 & 0.009265 & 0.002519 & 0.000993 & \textbf{0.000263} & 0.004796 & 0.002588 & 0.00222 & 0.004256 & 584 & 583 & 582 \\
chr6.pan.fa.a2fb268.2ff309f.1300c8a.smooth\_comp1\_rem\_nroot & 31.784529 & 842.285944 & 3.478898 & \textbf{2.974145} & 41.105871 & 11.802887 & 821.613878 & 34.297282 & 1555547 & 1555534 & 1555411 \\
chr6.pan.fa.a2fb268.2ff309f.1300c8a.smooth\_comp5\_nroot & 0.032999 & 0.010347 & 0.004685 & \textbf{0.001382} & 0.020802 & 0.01053 & 0.009827 & 0.019519 & 1896 & 1894 & 1894 \\
chr6.pan.fa.a2fb268.2ff309f.1300c8a.smooth\_comp6\_rem\_nroot & 0.009487 & 0.003526 & 0.001126 & \textbf{0.000329} & 0.009572 & 0.003558 & 0.003706 & 0.009728 & 586 & 585 & 584 \\
chrM.pan.fa.6626ff2.7748b33.72587dd.smooth & 0.00737 & 0.002175 & 0.000926 & \textbf{0.000205} & 0.004955 & 0.002541 & 0.002136 & 0.003888 & 456 & 456 & 456 \\
LPA & 0.022737 & 0.007408 & 0.002447 & \textbf{0.000644} & 0.01383 & 0.007358 & 0.006941 & 0.013945 & 1305 & 1305 & 1305 \\
lpa120 & 0.000644 & 0.000248 & 0.000104 & \textbf{0.000023} & 0.000609 & 0.000312 & 0.000276 & 0.000504 & 56 & 56 & 56 \\
MHC-57b\_VG & 0.009115 & 0.001483 & 0.000377 & \textbf{0.000158} & 0.003492 & 0.00108 & 0.001324 & 0.002066 & 376 & 376 & 376 \\
synth1 & 0.053129 & 0.725675 & 0.011024 & \textbf{0.001949} & 0.068957 & 0.019087 & 0.622382 & 0.080897 & 4662 & 4662 & 4453 \\
synth2 & 1.124042 & 12.468913 & 0.052957 & \textbf{0.032267} & 1.432235 & 0.06599 & 9.993535 & 1.186372 & 102585 & 102585 & 102421 \\
synth3 & 0.16155 & 0.132107 & 0.004828 & \textbf{0.003247} & 0.171983 & 0.006113 & 0.106613 & 0.140186 & 16542 & 16542 & 16536 \\
synth4 & 0.228025 & 0.16336 & 0.005784 & \textbf{0.00399} & 0.234178 & 0.007296 & 0.123888 & 0.174403 & 20810 & 20810 & 20804 \\
synth5 & 5.590897 & 60.755741 & 0.133024 & \textbf{0.125709} & 7.536372 & 0.178777 & 41.814505 & 5.537226 & 467294 & 467294 & 467110 \\
synth6 & 1.44153 & 0.721291 & 0.029062 & \textbf{0.025942} & 1.083951 & 0.028797 & 0.483111 & 0.760768 & 119602 & 119602 & 119601 \\
\bottomrule
\end{tabular}%
}
\end{table}

\begin{table}[ht!]
\centering
\caption{Total preprocessing execution times (in seconds) for each algorithm and dataset. See Tables~\ref{tab:algos} and \ref{tab:pre} for a breakdown of each algorithm's specific preprocessing steps and their individual execution times.}
\label{tab:pretotal}
\resizebox{\textwidth}{!}{%
\begin{tabular}{lccccccc}
\toprule
Dataset &
UltraLCA &
SWEEP &
\shortstack{SWEEP\\INTERVAL} &
\shortstack{HYBRID\\UltraLCA\\ONLINE} &
\shortstack{HYBRID\\SWEEP\\ONLINE} &
\shortstack{NESTED\\UltraLCA} &
\shortstack{NESTED\\HYBRID\\UltraLCA} \\
\midrule
C4-90\_VG\_rem\_nroot & 0.003 & 0.003 & 0.003 & 0.003 & 0.003 & 0.003 & 0.003 \\
cerevisiae.fa.gz.d1a145e.417fcdf.7493449.smooth.final\_comp11\_rem\_nroot & 1.308 & 1.433 & 1.375 & 1.308 & 1.433 & 2.060 & 2.060 \\
cerevisiae.fa.gz.d1a145e.417fcdf.7493449.smooth.final\_comp4\_nroot & 1.970 & 2.127 & 1.987 & 1.970 & 2.127 & 3.193 & 3.193 \\
cerevisiae.fa.gz.d1a145e.417fcdf.7493449.smooth.final\_comp7\_rem\_nroot & 6.746 & 7.572 & 7.321 & 6.746 & 7.572 & 10.890 & 10.890 \\
chr19.pan.fa.a2fb268.e820cd3.9ea71d8.smooth\_comp14\_rem\_nroot & 110.914 & 121.456 & 117.139 & 110.914 & 121.456 & 168.979 & 168.979 \\
chr6.C4 & 0.169 & 0.171 & 0.042 & 0.169 & 0.171 & 0.185 & 0.185 \\
chr6.pan.fa.a2fb268.2ff309f.1300c8a.smooth\_comp1\_rem\_nroot & 200.634 & 223.747 & 215.673 & 200.634 & 223.747 & 343.852 & 343.852 \\
chr6.pan.fa.a2fb268.2ff309f.1300c8a.smooth\_comp5\_nroot & 0.554 & 0.558 & 0.139 & 0.554 & 0.558 & 0.609 & 0.609 \\
chr6.pan.fa.a2fb268.2ff309f.1300c8a.smooth\_comp6\_rem\_nroot & 0.043 & 0.044 & 0.042 & 0.043 & 0.044 & 0.058 & 0.058 \\
chrM.pan.fa.6626ff2.7748b33.72587dd.smooth & 0.029 & 0.030 & 0.029 & 0.029 & 0.030 & 0.041 & 0.041 \\
LPA & 0.083 & 0.085 & 0.081 & 0.083 & 0.085 & 0.117 & 0.117 \\
lpa120 & 0.006 & 0.007 & 0.006 & 0.006 & 0.007 & 0.008 & 0.008 \\
MHC-57b\_VG & 0.051 & 0.052 & 0.051 & 0.051 & 0.052 & 0.060 & 0.060 \\
synth1 & 0.175 & 0.182 & 0.174 & 0.175 & 0.182 & 0.262 & 0.262 \\
synth2 & 5.258 & 5.511 & 5.323 & 5.258 & 5.511 & 18.879 & 18.879 \\
synth3 & 1.316 & 1.402 & 1.366 & 1.316 & 1.402 & 1.829 & 1.829 \\
synth4 & 1.130 & 1.241 & 1.200 & 1.130 & 1.241 & 1.786 & 1.786 \\
synth5 & 30.259 & 33.736 & 32.820 & 30.259 & 33.736 & 49.172 & 49.172 \\
synth6 & 9.100 & 10.259 & 10.007 & 9.100 & 10.259 & 14.244 & 14.244 \\
\bottomrule
\end{tabular}%
}
\end{table}

\section{Discussion and conclusions}
\label{sec:discussion}
In this work, we expanded our previous LCA-based method for identifying ultrabubbles in a given set of snarls within a rooted bipartite biedged graph. Specifically, we presented three distinct ideas for improving the LCA-based method: (i, sweep) use the discovery and finishing times from a DFS traversal of the graph's BFS tree to efficiently process the nodes from the $ftip$ set, (ii, hybrid) select between the UltraLCA and Naive algorithm depending on the size of the snarl in relation to the $ftip$ set, and (iii, nested) use the parent-child structure of nested snarls to skip processing of nested snarls. We develop these ideas into six different algorithms and benchmark these on a set of real and synthetic graphs. 

Whereas the hybrid and nested algorithms had execution times similar to or slightly faster than UltraLCA, the fastest of the two sweep algorithms, SWEEP\_INTERVAL, was up to 1000 times faster than UltraLCA. The reason is that SWEEP\_INTERVAL uses constant time per snarl, giving a linear asymptotic complexity of $O(n+m+K)$ compared with $O(Kn)$ for UltraLCA  \cite{zisis2026ultrabubble}, where $K$, $n$, and $m$ are the number of snarls, nodes, and edges, respectively.

Although the asymptotic complexities and total running times of our algorithms, especially the sweep-based related ones, are promising for more dense graphs and complex snarl structures, the practical gain will depend on the feasibility of integrating the preprocessing steps into approaches for identifying snarls.
Of the methods, the nested family algorithms would likely be the most difficult to integrate, since it requires knowing the full set of snarls.

\section*{Author contributions}

A.E.Z. conceived the methods, conducted the experiments, analyzed the results, and wrote the manuscript. P.S. reviewed the manuscript, contributed to methodological discussions, and supervised the work.

\bibliographystyle{splncs04}
\bibliography{sn-bibliography}

%\section{Appendix}
\appendix
\label{sec:appendix}

\section{Preprocessing}
\label{subsec:preprocessing}

In the following, we present the different preprocessing steps required for UltraLCA \cite{zisis2026ultrabubble} and the algorithms proposed here; Table~\ref{tab:algos} shows the steps used by each algorithm and their resulting asymptotic complexities. Table~\ref{tab:pre} shows the execution times for each preprocessing step on each graph in our benchmark.

\begin{table}[ht!]

\centering
\caption{Algorithms preprocessing dependencies and asymptotic complexities. The columns show the algorithm, the algorithm's different preprocessing steps, and the algorithm's asymptotic complexity. A check mark shows that the algorithm requires the corresponding preprocessing step; $^{*}$ notes that for SWEEP\_INTERVAL, only the \texttt{Rooted preorder tables} part of this stage is used as the \texttt{RMQ build} (RMQ structure) is not required by SWEEP\_INTERVAL. In the asymptotic complexities, $K$, $n$, and $m$ denote the number of candidate snarls, nodes, and edges in the graph. 
}
\label{tab:algos}
\resizebox{\textwidth}{!}{%
\begin{tabular}{lcccccccccccc}
\toprule
Method &
\shortstack{Graph\\build} &
\shortstack{Snarls\\parse} &
\shortstack{Snarl R--L\\normalization} &
\shortstack{Root\\check} &
\shortstack{LCA / preorder\\tables} &
\shortstack{FTIP\\set} &
\shortstack{FTIP preorder\\points} &
\shortstack{R--L preorder\\intervals} &
\shortstack{Prefix count\\table} &
\shortstack{DFS\\order} &
\shortstack{Nested family\\build} &
\shortstack{Asymptotic\\complexity} \\
\midrule
% NAIVE                 & \checkmark & \checkmark &            &            &            &            &            &            &            &            &            & $O(K(n+m))$ \\
UltraLCA                 & \checkmark & \checkmark & \checkmark & \checkmark & \checkmark & \checkmark &            &            &            &            &            & $O(Kn)$ \\
SWEEP                 & \checkmark & \checkmark & \checkmark & \checkmark & \checkmark & \checkmark & \checkmark & \checkmark &            &            &            & $O(Kn)$ \\
SWEEP\_INTERVAL       & \checkmark & \checkmark & \checkmark & \checkmark & \checkmark$^{*}$ & \checkmark & \checkmark & \checkmark & \checkmark &            &            & $O(n+m+K)$ \\
HYBRID\_UltraLCA\_ONLINE & \checkmark & \checkmark & \checkmark & \checkmark & \checkmark & \checkmark &            &            &            &            &            & $O(K(n+m)+Kn)$\;(\emph{Naive + UltraLCA}) \\
HYBRID\_SWEEP\_ONLINE & \checkmark & \checkmark & \checkmark & \checkmark & \checkmark & \checkmark & \checkmark & \checkmark &            &            &            & $O(K(n+m)+Kn)$\;(\emph{Naive + Sweep }) \\
NESTED\_UltraLCA         & \checkmark & \checkmark & \checkmark & \checkmark & \checkmark & \checkmark &            &            &            & \checkmark & \checkmark & $O(Kn)$ \\
NESTED\_HYBRID\_UltraLCA & \checkmark & \checkmark & \checkmark & \checkmark & \checkmark & \checkmark &            &            &            & \checkmark & \checkmark & $O(K(n+m)+Kn)$\;(\emph{Naive + UltraLCA }) \\
\bottomrule
\end{tabular}%
}
\end{table}

Note that the algorithms share six preprocessing steps, which in total have linear $O(n+m)$ complexity. Henceforth, We refer to these as the common subprocesses. When determining each algorithm's asymptotic complexity, we will consider the additional preprocessing steps needed by each algorithm.

\subsection{Graph build}
The "Graph build" process parses the given GFA file and (i) for each segment (node), $a$, creates a left node $a_L$ and a right node $a_R$ and connects them by a black edge, and (ii) for each link (edge), creates a corresponding gray edge. Note that the GFA graphs used here have intentionally no L-L or R-R edges as the graphs were only used for benchmarking the algorithms.
We therefore did not implement the linear time method of transforming a biedged graph into a biedged bipartite graph as described in \cite{zisis2026ultrabubble} and detailed in \cite{rautiainen2020graphaligner}. Nevertheless, the complexity of building the graph is $O(n+m)$, where $n$ and $m$ are the number of nodes and edges in the given variation graph.

\subsection{Snarls parse}\label{sec:snarlsparse}

The "Snarls parse" process parses the JSON output from the vg tool to extract all acyclic snarls of the given biedged graph $B$. Each snarl is represented as a tuple corresponding to the snarl's frontier nodes. The complexity is $O(K)$, where $K$ is the number of snarls, as each snarl is read and converted once. Note that at this stage, the snarls have not yet been characterized as $R-L$ or $L-R$ snarls; see Section~\ref{sec:rlnorm}. 

\subsection{Root check}\label{sec:rootcheck}
The "Root check" process checks that the graph has a working root; that is, a node such that when starting in that node one can traverse the whole graph by following the traversal rules for biedged graphs \cite{zisis2026ultrabubble}.
Specifically, in one modified BFS passage, the process checks that the provided root is an actual working root. The process therefore has linear $O(n+m)$ complexity. More information regarding how to build a working root if none exists, without altering the set of ultrabubbles of the graph, can be found in \cite{zisis2026ultrabubble}.

\subsection{LCA preorder tables}\label{sec:LCApreordertabs}
The "LCA preorder tables" process first constructs
a rooted tree $B_t$ by running a BFS from the root of the graph.  
Each node is assigned a parent and a depth, which shows the node's distance from the root. Second, the process runs a DFS in $B_t$ and assigns each node a preorder index and subtree end index, corresponding to the node's discovery time and fininshing time in the DFS. The DFS also records a Euler tour of $B_t$; that is, the sequence of nodes visited while moving from parents to children and back during the DFS. Third, a range minimum query (RMQ) is constructed based on the Euler tour and the depth values so that LCA queries in $B_t$ can be answered in $O(1)$ time. 
This process has complexity $(O(n+m)$, where $n$ and $m$ are the number of nodes and edges in the biedged graph.
The first two steps were implemented in \texttt{Rooted preorder tables} while the third step in \texttt{RMQ build} (see Table~\ref{tab:pre} for their specific execution times).

\subsection{Snarl R-L normalization}
\label{sec:rlnorm}
The "Snarl R-L normalization" process takes as input the set of snarls (Section~\ref{sec:snarlsparse}) and the BFS depth of each node (Section~\ref{sec:LCApreordertabs}) and for each snarl, identifies the frontier node that is closest to the root. 
In this way, all snarls are represented with respect to the rooted traversal of the graph. Moreover, all $R-L$ and $L-R$ snarls can be identified, as their frontier nodes will have unequal distances from the root. 
Since each candidate snarl is processed once and no new graph traversal is needed, this step has complexity $O(K)$, where $K$ is the number of candidate snarls. 
\subsection{Ftip set}\label{sec:ftipset}
The "Ftip set" process uses Tarjan's approach \cite{21} to identify the $ftip$ set of the given graph $B$. Briefly, this approach uses a modified DFS to identify cycle closing nodes while also identifying the tips in the graph; see \cite{zisis2026ultrabubble} for details. Note that as we are only considering acyclic snarls from the vg tool, the $ftip$ set will consist only of tips. In this benchmark, we are therefore only detecting the tips in the graph. This step has complexity $O(n+m)$, where $n$ and $m$ are the number of nodes and edges in the biedged graph.

\subsection{Ftip preorder points}\label{sec:ftipprepoints}
The "Ftip preorder points" process takes as input the $ftip$ set (Section~\ref{sec:ftipset}) and the preorder values of all nodes of the graph (Section~\ref{sec:LCApreordertabs}), scans all nodes of the graph in preorder order, and stores all nodes from the $ftip$ set in a list along with its preorder index.
This way, the $ftip$ set is ordered with respect to the preorder DFS traversal of $B_t$.
Since each node is checked once and each node $t \in ftip$ set is stored once, the process has complexity $O(n+F)$, where $n$ and $F$ are the number of nodes and the size of $ftip$ for the graph. Note that this complexity simplifies to $O(n)$, as \(F \le n\).

\subsection{R-L preorder intervals}
The "R-L preorder intervals" process takes as input the set of normalized snarls (Section~\ref{sec:rlnorm}) and the preorder and subtree end values of all nodes (Section~\ref{sec:LCApreordertabs}). 
Then, for each $R-L$ snarl in the input, the process stores the left and right frontier nodes, and the preorder and subtree end values of the snarl's left frontier node $sn_i$ as tuples in a bucket array indexed by the preorder value of $sn_i$. 
The process has complexity $O(n+K)$, where $n$ and $K$ are the number of nodes and snarls in the biedged graph.

\subsection{Prefix count table}
The "Prefix count table" process takes as input the $ftip$ preorder points (Section~\ref{sec:ftipprepoints}) and builds a prefix count table that for each preorder position of $B_t$ gives the number of nodes from the $ftip$ set that are located after that position in the DFS traversal. Specifically, the prefix count table is built by first creating an array indexed by the preorder positions of $B_t$ and that contains 1 in each position corresponding to a node in the $ftip$ set and 0 otherwise, and then computing the cumulative sum over the array in reverse. This process has complexity $O(n)$, where $n$ is the number of nodes in the biedged graph.

\subsection{DFS order}\label{sec:dfsorder}
The "DFS order" process performs a modified DFS in the biedged graph $B$ (instead of the BFS tree of $B$, as in Section~\ref{sec:LCApreordertabs}). 
During this DFS, each node is assigned a discovery position when it is first discovered and a finishing position which records the point that the DFS traversal leaves the node after all allowed successors have been explored. Moreover, for each node,  the highest reachable discovery index inside its DFS subtree is stored. This way, each node records when it was first visited, when its exploration finished, and the extent of its DFS subtree. While the DFS traversal it self has linear time complexity $O(n+m)$, in the current implementation, the allowed successors of each node are sorted before the DFS step, in order to keep the traversal order deterministic. This gives an overall running time of \(O(n + m + \sum_{v} d(v)\log d(v))\), where $n$, $m$, and $d(v)$ are the number of nodes, edges, and the degree of each node $v$ of the biedged graph, respectively. This simplifies to \(O(n + m\log n)\). In practice, the typical behavior is linear because the number of allowed successors of every node $v$ is relatively small.

\subsection{Nested family build}
The "Nested family build" process takes as input the frontier nodes of the $R-L$ snarls (Section~\ref{sec:rlnorm}) and the DFS information (Section~\ref{sec:dfsorder}).

First, for each snarl, it captures the DFS discovery positions of its two frontier nodes. Second, the snarls are ordered so that possible outer snarls are considered before inner ones. The parent-child relations are identified by a first passage that uses a simple DFS frontier containment. Specifically, a snarl is placed inside another when the left frontier of the outer snarl is discovered earlier and the right frontier of the inner snarl is discovered before the right frontier of the outer one. However, in some cases, the right frontier of an outer snarl may be discovered early through one branch of the graph, while a truly nested snarl is discovered later through another branch. For this reason, a fallback pass is also required. More specifically, for a candidate parent, a snarl is considered nested inside it, if its left frontier is discovered only after the DFS subtree of the parent’s right frontier has already ended, while both frontiers of the snarl still lie inside the DFS subtree of the parent’s left frontier. For each candidate parent examined in the fallback pass, the method tests the snarls that satisfy this fallback nesting condition with respect to that parent. If such a snarl has not yet been assigned to any parent, it is assigned to the current candidate parent. If it has already been assigned to another parent, but that parent is less deeply nested than the current one, then it is reassigned to the current candidate parent. Finally, the snarls that have no parent become family roots, and the process stores the parent-child relations, the depth of each snarl, and a top-down processing order inside each family, from the outer snarls to the inner ones. Regarding the complexity, the first part of the construction is $O(n+K)$, since it uses arrays over the $n$ DFS positions and processes the $K$ candidate snarls once. The fallback pass then performs additional candidate parent-child checks. If $A$ is the number of such fallback checks, the implemented running time is $O(n+K+A)$. In the worst case, $A$ can be $O(K^2)$, and thus the overall complexity becomes $O(n+K^2)$. However, in practice $A$ was much smaller and the observed behavior was close to linear.

The theoretical framework motivates a linear-time hierarchy characterization. In the present implementation, however, we used a more explicit DFS-based family-construction procedure with a fallback pass. This implementation was preferred because it is easier to validate and control within the current pipeline, while remaining faithful to the intended hierarchy logic. Although its worst-case complexity is less tight, the observed behavior in practice remained close to linear.

\begin{table}[ht!]
\centering
\caption{Subprocesses runtimes in seconds for each dataset. Note that \texttt{Rooted preorder tables} and \texttt{RMQ build} are the two parts of \texttt{LCA-preorder tables} process.}
\label{tab:pre}
\resizebox{\textwidth}{!}{%
\begin{tabular}{lcccccccccccc}
\toprule
Dataset & \shortstack{Graph\\build} & \shortstack{Snarls\\parse} & \shortstack{Snarl R--L\\norm.} & \shortstack{Root\\check} & \shortstack{Rooted\\preorder\\tables} & \shortstack{RMQ\\build} & \shortstack{FTIP\\set} & \shortstack{FTIP\\preorder\\points} & \shortstack{R--L\\preorder\\intervals} & \shortstack{Prefix\\count\\table} & \shortstack{DFS\\order} & \shortstack{Nested\\family\\build} \\
\midrule
C4-90\_VG\_rem\_nroot & 0.002259 & 0.000287 & 0.000014 & 0.000216 & 0.000089 & 0.000062 & 0.000017 & 0.000009 & 0.000016 & 0.000009 & 0.000256 & 0.000191 \\
cerevisiae.fa.gz.d1a145e.417fcdf.7493449.smooth.final\_comp11\_rem\_nroot & 0.302024 & 0.060620 & 0.020353 & 0.471956 & 0.358436 & 0.062722 & 0.031871 & 0.005388 & 0.119821 & 0.004574 & 0.502616 & 0.249742 \\
cerevisiae.fa.gz.d1a145e.417fcdf.7493449.smooth.final\_comp4\_nroot & 0.508440 & 0.073719 & 0.034992 & 0.646945 & 0.512214 & 0.146251 & 0.047253 & 0.007801 & 0.149417 & 0.006166 & 0.812958 & 0.410583 \\
cerevisiae.fa.gz.d1a145e.417fcdf.7493449.smooth.final\_comp7\_rem\_nroot & 1.747037 & 0.219589 & 0.143555 & 2.095786 & 2.129141 & 0.270054 & 0.140385 & 0.025033 & 0.801054 & 0.019249 & 2.714417 & 1.430482 \\
chr19.pan.fa.a2fb268.e820cd3.9ea71d8.smooth\_comp14\_rem\_nroot & 27.786487 & 2.878278 & 2.770625 & 34.157704 & 36.239541 & 4.595366 & 2.486461 & 0.293645 & 10.248089 & 0.277857 & 38.575769 & 19.488436 \\
chr6.C4 & 0.015783 & 0.002364 & 0.000594 & 0.013464 & 0.006973 & 0.128313 & 0.001253 & 0.000221 & 0.001588 & 0.000164 & 0.011953 & 0.004021 \\
chr6.pan.fa.a2fb268.2ff309f.1300c8a.smooth\_comp1\_rem\_nroot & 50.585821 & 4.906081 & 5.201421 & 59.454812 & 65.387433 & 8.522116 & 6.576316 & 0.754706 & 22.358571 & 0.447384 & 74.220116 & 68.997831 \\
chr6.pan.fa.a2fb268.2ff309f.1300c8a.smooth\_comp5\_nroot & 0.040581 & 0.008387 & 0.002372 & 0.048556 & 0.028008 & 0.422646 & 0.003261 & 0.000601 & 0.003137 & 0.003844 & 0.040595 & 0.014525 \\
chr6.pan.fa.a2fb268.2ff309f.1300c8a.smooth\_comp6\_rem\_nroot & 0.015108 & 0.005578 & 0.000282 & 0.012635 & 0.006121 & 0.002111 & 0.000740 & 0.000109 & 0.001275 & 0.000183 & 0.011504 & 0.004334 \\
chrM.pan.fa.6626ff2.7748b33.72587dd.smooth & 0.008571 & 0.001668 & 0.000254 & 0.011029 & 0.005292 & 0.001672 & 0.000607 & 0.000071 & 0.000904 & 0.000132 & 0.008555 & 0.003310 \\
LPA & 0.026222 & 0.007460 & 0.000932 & 0.028492 & 0.013967 & 0.004388 & 0.001806 & 0.000202 & 0.001758 & 0.000311 & 0.025149 & 0.009014 \\
lpa120 & 0.002377 & 0.001922 & 0.000031 & 0.001043 & 0.000428 & 0.000576 & 0.000058 & 0.000010 & 0.000059 & 0.000015 & 0.000775 & 0.000376 \\
MHC-57b\_VG & 0.035303 & 0.002171 & 0.000200 & 0.008232 & 0.003649 & 0.001170 & 0.000454 & 0.000066 & 0.000495 & 0.000095 & 0.006915 & 0.002309 \\
synth1 & 0.038361 & 0.016455 & 0.004475 & 0.066303 & 0.035432 & 0.009560 & 0.004335 & 0.000858 & 0.006632 & 0.000695 & 0.053493 & 0.033476 \\
synth2 & 0.929140 & 0.325011 & 0.255755 & 2.018099 & 1.391860 & 0.202751 & 0.135295 & 0.022307 & 0.230509 & 0.014940 & 10.988207 & 2.633255 \\
synth3 & 0.137992 & 0.053019 & 0.027653 & 0.268357 & 0.768812 & 0.038514 & 0.021280 & 0.003819 & 0.082872 & 0.002394 & 0.257276 & 0.255684 \\
synth4 & 0.310861 & 0.068278 & 0.038126 & 0.371289 & 0.270578 & 0.044325 & 0.026370 & 0.004558 & 0.106751 & 0.003173 & 0.349580 & 0.306860 \\
synth5 & 7.022819 & 1.445405 & 1.385201 & 9.818418 & 8.925493 & 0.989032 & 0.672348 & 0.083023 & 3.394742 & 0.072408 & 10.163069 & 8.749870 \\
synth6 & 2.956847 & 0.375953 & 0.313441 & 2.477405 & 2.543365 & 0.269731 & 0.163213 & 0.026715 & 1.132306 & 0.018031 & 2.688034 & 2.456199 \\
\bottomrule
\end{tabular}%
}
\end{table}

\section{Omitted proofs}\label{sec:omitted_proofs}
\begin{figure}[ht!]
\centering
\includegraphics[width=\textwidth]{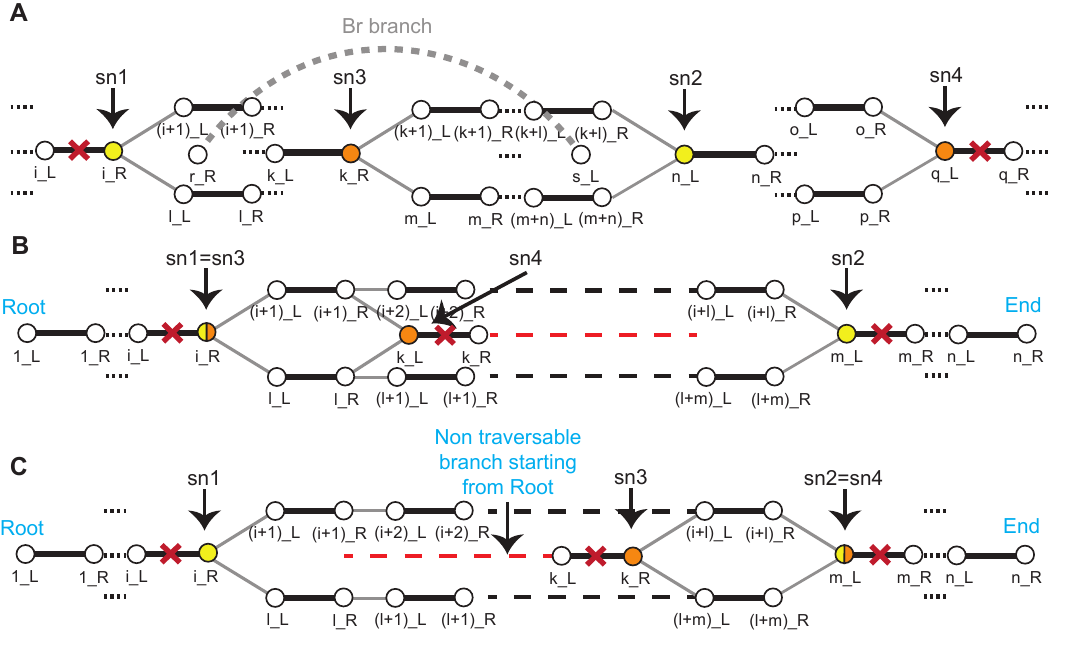}
\caption{Crossing snarls cannot co-exist (\textbf{A}), overlapping snarls having a common leftmost frontier node cannot form ultrabubbles (\textbf{B}), and there cannot be overlapping snarls that share their rightmost frontier node (\textbf{C}). (\textbf{A}) Snarls $(sn_1,sn_2)$ and $(sn_3,sn_4)$ cannot occur. First, if the branch $B_r$ does not exist then the edge $sn_3'-sn_3$ ($k\_L-k\_R$) violates the minimality criterion for snarl $(sn_1,sn_2)$. Second, if $B_r$ is present the separable criterion of snarl $(sn_3,sn_4)$ is violated since nodes $sn_3', sn_1, sn_2$ and $sn_4$ end up to the same component. (\textbf{B}) Overlapping snarls $(sn_1,sn_2)$ and $(sn_3,sn_4)$ that share their leftmost frontier node $(sn_1=sn_3)$ cannot form ultrabubbles. The subgraph represented by the dashed red line, because of the separable criterion of the snarls, has to be isolated from $sn_2$ and thus will end up to a tip or cycle. Therefore, snarl $(sn_1,sn_2)$ is not an ultrabubble as it includes the aforementioned branch. Neither is snarl $(sn_3,sn_4)$, because the sink-end of the graph belongs to its subgraph and therefore also contains a tip or a cycle. (\textbf{C}) Overlapping snarls $(sn_1,sn_2)$ and $(sn_3,sn_4)$ that share their rightmost frontier node $(sn_2=sn_4)$ cannot occur. For this to happen, the subgraph represented by the red dashed line must be isolated from the rest of the graph because of the snarl separability criterion. However, as $B$ has a root, there has to be a path from the root to $sn_3$, which is a contradiction.
}
\label{fig:comp}
\end{figure}
\begin{restatedlemma}{lem:overlap}
In the set $S_{RL}$ of $R-L$ snarls $(x,y)$ of a given rooted biedged graph $B$, no pair of crossing snarls can exist.
\end{restatedlemma}

\begin{proof}
We assume that two $R-L$ snarls of a biedged graph $B$ with frontier nodes $(sn_1, sn_2)$ and $(sn_3, sn_4)$ are crossing; that is, they have some overlap without having any common frontier nodes. Assume, without loss of generality, that $sn_3$ is located between $sn_1$ and $sn_2$, while $sn_4$ is more to the right of $sn_2$. Then, 
starting from $sn_1$ we can reach $sn_4$ by two possible ways as described below.
First, $sn_4$ can only be reached from $sn_1$ by using the edge $sn_3'$-$sn_3$. 
Since we have assumed that $sn_3,sn_4$ is a snarl and thus follows the separability criterion, the edge $sn_3'-sn_3$ is such that $(sn_1,sn_3')$ and $(sn_3,sn_4)$ are following the separability criterion of snarls. Consequently, $sn_3'-sn_3$ is a bridge edge, which violates the minimality criterion for $(sn_1, sn_2)$ (see Figure~\ref{fig:comp}\textbf{A}).
Second, $sn_4$ can be reached from $sn_1$ by using another branch $Br$ and avoiding the edge $sn_3'-sn_3$.
Then, after deleting the black edges $sn_3'-sn_3$ and $sn_4-sn_4'$, because of the branch $Br$,  
nodes $sn_3'$, $sn_1$, $sn_2$ and $sn_4$ will still belong to the same component (see Figure~\ref{fig:comp}\textbf{A}).  This contradicts the assumption of $(sn_3,sn_4)$ being a snarl, as $sn_4$ and $sn_3'$ are in the same component, which violates the separability criterion.
\end{proof}

\begin{restatedlemma}{lem:leftmost} 
Given the set $S_{RL}$ of $R-L$ snarls $(x,y)$ of a rooted biedged graph $B$, any snarls in $S_{RL}$ that share the leftmost frontier node cannot be ultrabubbles.
\end{restatedlemma}

\begin{proof}
 Assume that we have a pair of $R-L$ snarls $(sn_1, sn_2)$ and $(sn_3, sn_4)$ that share their leftmost frontier node; that is, $sn_1=sn_3$ and $sn_2 \ne sn_4$ (see Figure~\ref{fig:comp}\textbf{B}). 
 First, consider snarl $(sn_3, sn_4)$ and the components we get after deleting the black edges $sn_3'-sn_3$ (which is identical to $sn_1'-sn_1$) and $sn_4-sn_4'$. 
 Then, the subgraph that includes $sn_4'$ cannot be connected to the subgraph of the snarl $(sn_1, sn_2)$, since this will place nodes $sn4'$ and $sn_3$ in the same component, thus violating the separability criterion for the snarl ($sn_3, sn_4$). Consequently, the path including $sn_4'$ will end in a different part in the graph.

 Moreover, the snarl ($sn_3, sn_4$) must include the edge $sn_2-sn_2'$ and the subgraph connected through this edge, which will have to include the end of the graph and therefore also a cycle or a tip. Snarl ($sn_3, sn_4$) is therefore not an ultrabubble.
 Second, consider snarl $(sn_1, sn_2)$. After deleting the black edges $sn_1'-sn_1$ and $sn_2-sn_2'$, the snarl subgraph $(sn_1, sn_2)$ will contain $sn_4-sn_4'$ and thereby also contain a tip or a cycle corresponding to the end of that path. Snarl $(sn_1, sn_2)$ is therefore not an ultrabubble (Figure~\ref{fig:comp}\textbf{B}).
\end{proof}

\begin{restatedlemma}{lem:rightmost}
Given the set $S_{RL}$ of $R-L$ snarls $(x,y)$ of a rooted biedged graph $B$, no two distinct snarls in $S_{RL}$ will share their rightmost frontier node $y$. 
\end{restatedlemma}
 \begin{proof}
 Assume that we have a pair of $R-L$ snarls $(sn_1, sn_2)$ and $(sn_3, sn_4)$ that share their rightmost frontier node; that is, $sn_2=sn_4$ and $sn_1 \ne sn_3$ (see Figure~\ref{fig:comp}\textbf{C}). First, consider $(sn_3, sn_4)$.
As in the proof for Lemma~\ref{lem:leftmost}, we can conclude that after removing the edge $sn_3'$-$sn_3$, the subgraph containing $sn_3'$ cannot be connected to the $(sn_1, sn_2)$ subgraph. Moreover, $sn_3'$ cannot be connected to the root, as this would also violate the separability criterion, as $sn1$, by the requirement of $B$ being rooted, must be connected to the root. However, by the requirement of $B$ being rooted so that the whole graph can be traversed from the root, there must exist a path from the root to $sn_3'$, which means that if $(sn_1, sn_2)$ is a snarl, then $(sn_3, sn_4)$ cannot be a snarl. The same logic applies to $(sn_1, sn_2)$.
 \end{proof}

\begin{restatedlemma}{lem:nested_finish}
Given a DFS traversal of the BFS tree of a rooted biedged graph $B$, the corresponding finishing times $t_f(v)$ for each node $v \in B$, the set $K_{c}$ of compatible $R-L$ snarls of $B$, and two snarls $\{(r_i, l_i), (r_j, l_j)\} \in K_c, i\ne j$, $(r_j, l_j)$ is nested within $(r_i, l_i)$ iff $t_f(r_i) > t_f(r_j) > t_f(l_j) > t_f(l_i)$. 
\end{restatedlemma}
\begin{proof}
First, if $(r_j, l_j)$ is nested within $(r_i, l_i)$, then $t_f(r_i) > t_f(r_j) > t_f(l_j)$ \cite{zisis2026ultrabubble}. This is because both $r_j$ and $l_j$ must be proper descendants of $r_i$ in the DFS forest, and because $l_j$ must be a proper descendant of $r_j$. Second, as $l_i$ must be a proper descendant of $r_i$, $t_f(r_i) > t_f(l_i)$. Third, as $(r_i, l_i)$ and $(r_j, l_j)$ are compatible, by Lemma~\ref{lem:overlap}, $l_i$ cannot be located within $(r_j, l_j)$, so either $t_f(l_i) > t_f(r_j)$ or $t_f(l_j) > t_f(l_i)$. 
Assume that $(r_j, l_j)$ is not nested within $(r_i, l_i)$. Since we already have that $t_f(r_i) > t_f(r_j) > t_f(l_j)$, this means that either $(r_j, l_j)$ is a proper descendant of $(r_i, l_i)$ in the DFS forest or $(r_i, l_i)$ is discovered after $(r_j, l_j)$ in the DFS. In both cases, $t_f(l_i) > t_f(r_j)$. Now, assume that $(r_j, l_j)$ is nested within $(r_i, l_i)$ and that $t_f(l_i) > t_f(r_j)$. This means that the DFS finishes the subtree containing $(r_j, l_j)$ before finishing $l_j$, which can only happen if there is no path from $r_j$ to $l_i$. However, if there is no path from $r_j$ to $l_i$, there is also no path from $l_j$ to $l_i$, which means that if $(r_j, l_j)$ is a snarl, it can be completely separated from $(r_i, l_i)$. Consequently, either $(r_i, l_i)$ is not minimal (and therefore no snarl) or $(r_j, l_j)$ is not nested within $(r_i, l_i)$, both of which contradict the assumptions. Hence, $(r_j, l_j)$ can only be nested within $(r_i, l_i)$ if $t_f(l_j) > t_f(l_i)$, which, combined with the first three observations, gives the inequality.
\end{proof}

\end{document}